\documentclass[a4paper,11pt]{article}
\usepackage{latexsym,amsmath,amssymb,enumerate,mathrsfs,amsthm}
\usepackage{graphicx}

\newcommand{\alg}[1]{\mathfrak{#1}}

\newcommand{\mc}[1]{\mathcal{#1}}

\newcommand{\linspan}{\operatorname{span}}
\newcommand{\bonds}{{\bf B}}
\newcommand{\set}[1]{\left\{ #1 \right\}}
\newcommand{\norm}[1]{\lVert#1\rVert}
\newcommand{\supp}{\operatorname{supp}}
\newcommand{\str}{\operatorname{star}}
\newcommand{\plaq}{\operatorname{plaq}}

\numberwithin{equation}{section}

\theoremstyle{plain}
\newtheorem{theorem}{Theorem}
\numberwithin{theorem}{section}
\newtheorem{lemma}[theorem]{Lemma}

\newtheorem{corollary}[theorem]{Corollary}

\newtheorem{definition}[theorem]{Definition}

\title{Haag duality and the distal split property for cones in the toric code}
\author{Pieter Naaijkens\footnote{Present address: Institut f\"ur Theoretische Physik, Leibniz Universit\"at Hannover, Germany. E-mail: \texttt{pieter.naaijkens@itp.uni-hannover.de}}\\%
Institute for Mathematics, Astrophysics and Particle Physics\\ Radboud University Nijmegen, The Netherlands}
\begin{document}
\maketitle

\begin{abstract}
\noindent We prove that Haag duality holds for cones in the toric code model. That is, for a cone $\Lambda$, the algebra $\mc{R}_{\Lambda}$ of observables localized in $\Lambda$ and the algebra $\mc{R}_{\Lambda^c}$ of observables localized in the complement $\Lambda^c$ generate each other's commutant as von Neumann algebras. Moreover, we show that the distal split property holds: if $\Lambda_1 \subset \Lambda_2$ are two cones whose boundaries are well separated, there is a Type I factor $\mc{N}$ such that $\mc{R}_{\Lambda_1} \subset \mc{N} \subset \mc{R}_{\Lambda_2}$. We demonstrate this by explicitly constructing~$\mc{N}$.
\end{abstract}

\section{Introduction}
For a finite group $G$, Kitaev introduced a quantum mechanical model with excitations described by the representation theory of a certain Hopf algebra, the quantum double of $G$~\cite{MR1951039}. In recent work we studied the superselection structure of the toric code model (where $G=\mathbb{Z}_2$) considered on a plane~\cite{toricendo}, in the spirit of the Doplicher-Haag-Roberts (DHR) program in algebraic quantum field theory~\cite{MR1405610}. Haag duality is an important tool in the DHR analysis, and although it is not necessary for the study of the superselection structure of the toric code~\cite{toricendo}, it does make the analysis more elegant. In the model we consider here, the appropriate formulation is as follows. Consider a cone-like region $\Lambda$, and write $\mc{R}_{\Lambda}$ for the von Neumann algebra generated by the observables localized in $\Lambda$ (in the GNS representation obtained from the ground state). One can then consider all observables localized in the complement $\Lambda^c$ of $\Lambda$, generating an algebra $\mc{R}_{\Lambda^c}$. By locality, i.e. the property that observables localized in disjoint regions commute, one has the inclusion $\mc{R}_{\Lambda^c} \subset \mc{R}_{\Lambda}'$, where the prime denotes the commutant. Haag duality is the statement that the reverse inclusion also holds. 

The \emph{distal split property}, a consequence of Haag duality, is perhaps of greater interest in the present context. The property says that if $\Lambda_1 \subset \Lambda_2$ are two cones whose boundaries are sufficiently well separated, then there is a Type I factor $\mc{R}_{\Lambda_1} \subset \mc{N} \subset \mc{R}_{\Lambda_2}$~\cite{toricendo}. The split property has been studied in a general operator algebraic framework~\cite{MR735338} and has important consequences in the context of algebraic quantum field theory (see e.g. \cite{MR848392}). In this note we present a new proof of the distal split property for the toric code by explicitly constructing an appropriate Type I factor $\mc{N}$.

The distal split property can be interpreted as a strong statistical independence of the regions $\Lambda_1$ and $\Lambda_2^c$. For if it holds, and if normal states $\varphi_1$ (resp. $\varphi_2$) of $\mc{R}_{\Lambda_1}$ (resp. $\mc{R}_{\Lambda_2}'$) are given, then there is a normal state $\varphi$ of $\mc{R}_{\Lambda_1} \vee \mc{R}_{\Lambda_2}'$ such that $\varphi(AB) = \varphi_1(A) \varphi_2(B)$. In other words, one can prepare a state in the region $\Lambda_1$ independently of the state in $\Lambda_2^c$. It is instructive to consider the relation with entanglement. If $\Lambda_1 \subset \Lambda_2$ are two cones whose boundaries are sufficiently well separated, then by the distal split property there are normal product states $\varphi$ as above. For such states there is no violation of Bell's inequalities for the pair $\mc{R}_{\Lambda_1}$, $\mc{R}_{\Lambda_2}$ of observable algebras (see~\cite{MR887998} for a precise formulation of Bell's inequalities in this context). On the other hand, if we choose $\Lambda_1 = \Lambda_2$, then $\mc{R}_{\Lambda_1}$ and $\mc{R}_{\Lambda_2^c}$ are maximally correlated (cf.~\cite{MR887998}). In fact there is \emph{infinite one-copy entanglement} between $\mc{R}_{\Lambda_1}$ and $\mc{R}_{\Lambda_2}$~\cite[Cor. 5.1]{MR2281418}, since $\mc{R}_{\Lambda_1}$ is not of Type I~\cite[Thm. 5.1]{toricendo} and Haag duality holds. It should be noted that the requirement on the separation of the boundaries is actually very weak: often a distance of one is already good enough. In other words, even a small shift of the cone $\Lambda_1$ inside $\Lambda_2$ can have great consequences for the entanglement properties.

As far as the author is aware, currently no general conditions implying Haag duality are known. However, there are proofs in specific cases, for example for certain quantum spin chain models~\cite{MR2281418,MR2605849} or in the setting of algebraic quantum field theory~\cite{MR0438944,MR1147468}. The proofs in the quantum spin chain case make use of the split property, a stronger condition than the \emph{distal} split property we consider in this note. This stronger split property does not hold in the model under consideration, since $\mc{R}_{\Lambda}$ is not of Type I if $\Lambda$ is a cone and we have Haag duality.

In studying commutation problems of von Neumann algebras, such as Haag duality, a natural tool is Tomita--Takesaki modular theory. In algebraic quantum field theory this theory is relevant because of the Reeh-Schlieder Theorem, according to which the vacuum vector is cyclic and separating for the observables localized in a double cone, i.e., the intersection of a forward and backward light cone. Indeed, this has been used to prove duality results, e.g. in~\cite{MR0438944,MR1147468}. In contrast, in the model we are considering, the ground state vector $\Omega$ is not cyclic for the algebra of observables localized in a cone, hence we cannot directly apply these techniques. Our strategy, therefore, is to restrict the algebras to a subspace $\mc{H}_\Lambda$ of the representation space $\mc{H}$, such that $\Omega$ \emph{is} cyclic for (the restriction of) $\mc{R}_\Lambda$. One can also restrict $\mc{R}_{\Lambda^c}$ to this subspace, and using a theorem of Rieffel and van Daele~\cite{MR0383096} one can prove that these restrictions generate each other's commutant as subalgebras of $\alg{B}(\mc{H}_\Lambda)$. The final step is to extend this to the algebras acting on $\mc{H}$. It turns out that similar techniques can be used to prove the distal split property.

In the next section we recall the toric code model as considered on a plane, and fix our notations. Section~\ref{sec:haag} contains a proof of the main result: Haag duality for cones. In the last section, the distal split property is shown to hold by constructing an interpolating Type I factor \emph{explicitly}, in contrast with results in algebraic quantum field theory where the existence follows from abstract arguments.

\section{The model}
\label{sec:model}
We first recall the main features of Kitaev's toric code model~\cite{MR1951039}, considered in the $C^*$-algebraic framework for quantum lattice systems~\cite{MR2345476,toricendo}. Consider a square $\mathbb{Z}^2$ lattice. On each bond of the lattice (an edge between two vertices of distance 1), there is a spin-1/2 degree of freedom. That is, at each bond $b$ the local state space is $\mc{H}_{\set{b}} = \mathbb{C}^2$, with observables $\alg{A}(\{b\}) = M_2(\mathbb{C})$. The set of bonds will be denoted by $\bonds$. If $\Lambda \subset \bonds$ is a finite set, $\alg{A}(\Lambda)$ is the algebra of observables living on the bonds of $\Lambda$. It is the tensor product of the observable algebras acting on the individual bonds of $\Lambda$. If $\Lambda_1 \subset \Lambda_2$ there is an obvious inclusion of corresponding algebras, obtained by identifying $\mc{H}_{\Lambda_2} \cong \mc{H}_{\Lambda_1} \otimes \mc{H}_{\Lambda_2 \setminus \Lambda_1}$. This defines a local net of algebras with respect to the inclusion $\alg{A}(\Lambda_1) \hookrightarrow \alg{A}(\Lambda_2)$ for $\Lambda_1 \subset \Lambda_2$. Define the algebra of local observables,
\[
	\alg{A}_{loc} = \bigcup_{\Lambda_f \subset \bonds}  \alg{A}(\Lambda_f),
\]
where the union is over the finite subsets $\Lambda_f$ of $\bonds$. The algebra $\alg{A}$ of quasi-local observables is the completion of $\alg{A}_{loc}$ in the norm topology, turning it into a $C^*$-algebra. Equivalently, one can see it as the inductive limit of the net $\Lambda \mapsto \alg{A}(\Lambda)$ in the category of $C^*$-algebras. Note that $\alg{A}$ is a uniformly hyperfinite (UHF) algebra~\cite{MR887100}. The algebra of observables localized in an arbitrary subset $\Lambda$ of $\bonds$ is defined as
\[
\alg{A}(\Lambda) = \overline{\bigcup_{\Lambda_f \subset \Lambda} \alg{A}(\Lambda_f)}^{\norm{\cdot}},
\]
where the union is again over finite subsets. An operator $A$ is said to have support in $\Lambda$, or to be localized in $\Lambda$, if $A \in \alg{A}(\Lambda)$. The set $\supp(A) \subset \bonds$ is the smallest subset in which $A$ is localized. 

\begin{figure}
  \begin{center}
  \includegraphics{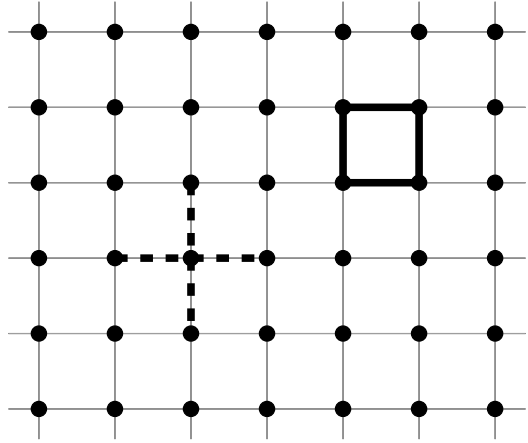}
\end{center}
\caption{The $\mathbb{Z}^2$ lattice. The gray bonds each carry a spin-1/2 degree of freedom. A star (dashed lines) and plaquette (thick lines) are shown.}
\label{fig:lattice}
\end{figure}
The Hamiltonian of Kitaev's model is defined in terms of plaquette and star operators, each supported on four bonds (see Figure~\ref{fig:lattice}). If $s$ is a point on the lattice, $\str(s)$ denotes the star based at $s$. Similarly, $\plaq(p)$ is the set of bonds enclosing a plaquette $p$. The corresponding star and plaquette operators are given by
\[
A_s = \bigotimes_{j \in \str(s)} \sigma_j^x, \qquad B_p = \bigotimes_{j \in \plaq(p)} \sigma_j^z,
\]
where the tensor product is understood as having Pauli matrices $\sigma^x$ (resp. $\sigma^z$) in places $j$, and unit operators in all other positions. It is then straightforward to check that for all stars $s$ and plaquettes $p$, we have
\[
	[A_s, B_p] = 0.
\]
These operators are used to define the local Hamiltonians. If $\Lambda_f \subset \bonds$ is finite, the associated local Hamiltonian is
\[
H_{\Lambda_f} = -\sum_{\str(s) \subset \Lambda_f}  A_s -\sum_{\plaq(p) \subset \Lambda_f} B_p.
\]
The model with dynamics described by these Hamiltonians has a unique ground state $\omega$,  and in the corresponding GNS representation the dynamics is implemented by a Hamiltonian with gap~\cite{MR2345476,toricendo}. This ground state is determined by the condition $\omega(A_s) = \omega(B_p) = 1$ for any star (resp. plaquette) operator $A_s$ (resp. $B_p$). The following Lemma, can be used to compute the value of the ground state on other operators.

\begin{lemma}
	\label{lem:state}
	Let $\omega$ be a state on a $C^*$-algebra $\alg{A}$, and suppose $X=X^*$ such that $X \leq I$ and $\omega(X) = 1$. Then $\omega(XY) = \omega(YX) = \omega(Y)$ for any $Y \in \alg{A}$.
\end{lemma}
This lemma follows from the Cauchy-Schwarz inequality (see~\cite[\S 2.1.1]{MR2345476} for a proof).

We write $(\pi, \Omega, \mc{H})$ for the GNS representation obtained from the ground state $\omega$. An easy calculation shows that $\omega( (A_s-I)^*(A_s-I) ) = 0$ for any star $s$. A similar result holds for the plaquette operators $B_p$, hence
\begin{equation}
	\pi(A_s) \Omega = \Omega, \quad \pi(B_p) \Omega = \Omega.
	\label{eq:stplaq}
\end{equation}
This relation will be useful later.

\begin{figure}
  \begin{center}
  \includegraphics{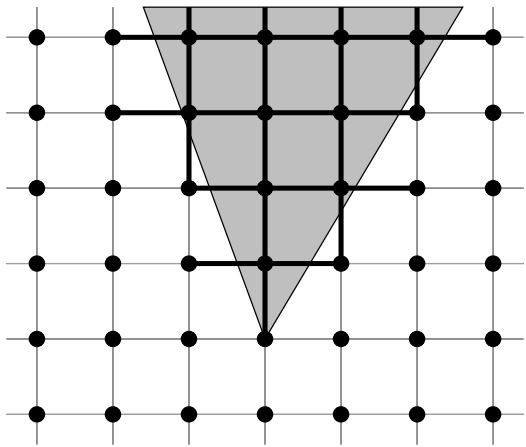}
\end{center}
\caption{Example of a cone (bold bonds). The shaded region is the area bounded by two rays emanating from a point.}
\label{fig:cone}
\end{figure}
We are mainly interested in (quasi)local observables localized in certain unbounded cone-like regions. An example is provided in Figure~\ref{fig:cone}. The precise definition is given below.
\begin{definition}
	Consider a point on the lattice $\mathbb{Z}^2$, with two rays emanating from it, such that the angle between those rays is positive but smaller than $\pi$. These two rays bound a convex subset of $\mathbb{R}^2$. A \emph{cone} $\Lambda \subset \bonds$ consists of all bonds that intersect the \emph{interior} of this convex area. 
\end{definition}

Next we consider paths. Let $x,y$ be two points in the lattice $\mathbb{Z}^2$. One can consider finite paths consisting of bonds between these points. Similarly, one can consider two plaquettes, or equivalently, two points on the dual lattice. A \emph{dual path} is a path on the dual lattice between two points. We identify such a dual path $\widehat{\xi}$ with the bonds $b \in \bonds$ that are crossed by this dual path. Corresponding with such paths there are string operators.\footnote{Note that we use notation different from Ref.~\cite{toricendo}.}

\begin{definition}
Suppose $\xi$ (resp. $\widehat{\xi}$) is a finite path on the lattice (resp. dual lattice). We define the corresponding \emph{string operators} by
\[
F_\xi := \bigotimes_{i \in \xi} \sigma_i^z, \quad F_{\widehat{\xi}} := \bigotimes_{i \in \widehat{\xi}} \sigma_i^x.
\]
\end{definition}
We will usually not distinguish between string operators corresponding to paths and those corresponding to dual paths. Note that by the properties of Pauli matrices, it is clear that string operators are self-adjoint, and that if $F_1, F_2$ are string operators, they either commute or anti-commute.

Now suppose that $\xi$ is a path that does not intersect itself. Then one sees that $F_\xi$ commutes with all star operators $A_s$, except for those corresponding to the star based at the endpoints of $\xi$. Clearly $F_\xi$ commutes with all plaquette operators. Considering the definition of the local Hamiltonians, $F_\xi \Omega$ can be interpreted as a state vector describing a pair of excitations at the endpoints of $\xi$. A similar argument holds for paths on the dual lattice, where the excitations are located at plaquettes, and we have anti-commutation with the corresponding plaquette operators.

Recall that if $\xi$ is a closed path, the corresponding operator $F_\xi$ can be written as a product of plaquette operators~\cite{toricendo}, hence $\omega(F_\xi) = 1$ by Lemma~\ref{lem:state}. Similarly, if $\xi$ is a closed dual path, $F_\xi$ is a product of star operators. From this it follows that $\pi(F_\xi) \Omega = \Omega$ for closed paths $\xi$. As an easy consequence, consider  two paths $\xi$ and $\xi'$ with the same endpoints. Then we have $\omega( (F_{\xi}-F_{\xi'})^*(F_{\xi}-F_{\xi'})) = 0$, because the cross-term $F_{\xi}F_{\xi'}$ is precisely the string operator corresponding to the loop formed by $\xi$ and $\xi'$. Hence $\pi(F_{\xi}) \Omega = \pi(F_{\xi'}) \Omega$. In physical terms this means that the excitations created do not depend on the path $\xi$, but only on its endpoints.

As we will see later, we will investigate excitations that appear near the edges of a cone $\Lambda$. Recall that a cone is described by two rays. These lines allow us to define what exactly is the boundary of a cone.
\begin{definition}
A vertex $v$ lies on the boundary of $\Lambda$ if and only if either $v$ lies on one of the two rays or $v$ lies outside the convex area bounded by the two rays and is one of the endpoints of a bond $b \in \Lambda$. A plaquette $p$ is at the boundary of $\Lambda$ if and only if some, but not all, bonds that enclose the plaquette are contained in $\Lambda$. The boundary of the complement $\Lambda^c$ of a cone is defined to be equal to the boundary of $\Lambda$.
\end{definition}

\section{Haag duality}\label{sec:haag}
Recall that $\pi$ is the GNS representation defined by the ground state. Suppose that $\Lambda$ is a cone. We can consider the von Neumann algebra generated by the observables localized in this cone, $\mc{R}_\Lambda := \pi(\alg{A}(\Lambda))''$, and similarly the algebra $\mc{R}_{\Lambda^c} := \pi(\alg{A}(\Lambda^c))''$ generated by observables localized in the complement of $\Lambda$. From locality it follows that $\mc{R}_\Lambda \subset \mc{R}_{\Lambda^c}'$. \emph{Haag duality} is the statement that the reverse inclusion is also true, i.e.
\begin{equation}
	\pi(\alg{A}(\Lambda))'' = \pi(\alg{A}(\Lambda^c))'.
	\label{eq:haagd}
\end{equation}
Our main result is that this is the case for the toric code model.
\begin{theorem}
	\label{thm:dual}
Let $\Lambda$ be a cone. Then in the ground state representation we have Haag duality, $\pi(\alg{A}(\Lambda))'' = \pi(\alg{A}(\Lambda^c))'$.
\end{theorem}

The basic idea behind the proof is to first reduce the problem to one of algebras acting on a Hilbert space $\mc{H}_\Lambda \subset \mc{H}$. This Hilbert subspace can be interpreted as the space of states with excitations localized in $\Lambda$. Before we proceed, first note that $\pi$ is a representation of a UHF (and hence simple) algebra, from which it is clear that $\pi$ is faithful. This makes it possible to identify $\pi(A)$ with $A$, for $A \in \alg{A}$, and we will do so from now on.

\begin{definition}
	Let $\Lambda$ be a cone. If $\xi$ is a path on the lattice, we say that it is contained in $\Lambda$ if $\xi \subset \Lambda$. A path $\xi$ on the dual lattice is contained in $\Lambda$ if each bond that intersects the dual path is in $\Lambda$. With this convention, we define 
\[	
	\mc{F}_\Lambda = \{ F_\xi : \xi\textrm{ is a path (or dual path) in } \Lambda \},
\]
and similarly for $\mc{F}_{\Lambda^c}$.
\end{definition}
The operators in $\mc{F}_\Lambda$ create excitations in $\Lambda$. Since $\Lambda \cup \Lambda^c = \bonds$, one would expect that the operators in $\mc{F}_\Lambda$ and $\mc{F}_{\Lambda^c}$ generate $\mc{H}$ by acting on the ground state vector $\Omega$. This is indeed the case:
\begin{lemma}
	\label{lem:denseset}
The closure of $\linspan \{ F_1 \cdots F_m \widehat{F}_1 \cdots \widehat{F}_n \Omega : F_i \in \mc{F}_\Lambda, \widehat{F}_j \in \mc{F}_{\Lambda^c} \}$ is equal to $\mc{H}$.
\end{lemma}
\begin{proof}
	Let $b \in \bonds$ and consider the path $\xi = \{ b \}$ and the dual path $\widehat{\xi}$ of length one crossing this bond. Then $I, F_\xi, F_{\widehat{\xi}}$ and $F_\xi F_{\widehat{\xi}}$ span the algebra $M_2(\mathbb{C})$ acting on this bond. By considering more bonds, one sees that all local operators can be obtained in this way, from which the statement follows since the local operators are dense in $\alg{A}$, and $\Omega$ is cyclic for $\pi(\alg{A})$ by the GNS construction.
\end{proof}

Next we consider the Hilbert space of all excitations localized in $\Lambda$.
\begin{definition}
	Consider the closure of $\linspan \{F_1 \cdots F_k \Omega : F_i \in \mc{F}_\Lambda \}$ and let $P_\Lambda$ be the projection onto this subspace of $\mc{H}$. We write $\mc{H}_\Lambda$ for the Hilbert space $\mc{H}_\Lambda = P_\Lambda \mc{H}$.
\end{definition}

\begin{lemma}
	\label{lem:unique}
	We have $\alg{A}(\Lambda) \mc{H}_\Lambda \subset \mc{H}_\Lambda$. In fact, $A \in \alg{A}(\Lambda)''$ is completely determined by its restriction to $\mc{H}_\Lambda$. 
\end{lemma}
\begin{proof}
	The algebra $\alg{A}(\Lambda)_{loc}$ is generated by operators $F_\xi$ for paths (and dual paths) $\xi$ contained in $\Lambda$. Such operators clearly map the linear subspace spanned by vectors of the form $F_1 \cdots F_k \Omega$ ($F_i \in \mc{F}_\Lambda$) into itself. Since this space is dense in $\mc{H}_\Lambda$, and $\alg{A}(\Lambda)_{loc}$ is dense in $\alg{A}(\Lambda)$, the first claim follows.

The second claim follows from the fact that if $AB=0$ for $A \in \mc{R}$ with $\mc{R}$ a factor, and $B \in \mc{R}'$, then either $A$ or $B$ is zero~\cite[Thm. 5.5.4]{MR719020}. Since $\alg{A}(\Lambda)''$ is a factor~\cite{toricendo} and $P_\Lambda \in \alg{A}(\Lambda)'$ by the previous part, the result follows. There is also an easy direct proof. We give it here since we will use a similar argument later on. Let $A_1, A_2 \in \alg{A}(\Lambda)$ and suppose that $A_1 \xi = A_2\xi$ for every $\xi \in \mc{H}_\Lambda$. Now consider $\eta =\widehat{F}_1 \cdots \widehat{F}_m F_1 \cdots F_n \Omega \in \mc{H}$, where again $F_i \in \mc{F}_\Lambda$ and $\widehat{F}_j \in \mc{F}_{\Lambda^c}$. Then we have
\[
A_1 \eta = \widehat{F}_1 \cdots \widehat{F}_m A_1  F_1 \cdots F_n \Omega  = \widehat{F}_1 \cdots \widehat{F}_m A_2  F_1 \cdots F_n \Omega = A_2 \eta.
\]
Since vectors of this form are dense in $\mc{H}$, the claim follows. If $A \in \alg{A}(\Lambda)''$, the statement follows in precisely the same way, since by locality we have $\alg{A}(\Lambda)'' \subset \alg{A}(\Lambda^c)'$.
\end{proof}

Consider now the algebra $\alg{A}(\Lambda^c)$ of observables localized in the complement of $\Lambda$. We want to show Haag duality, i.e. equation~\eqref{eq:haagd}, so $\alg{A}(\Lambda^c)'$ should map $\mc{H}_\Lambda$ into itself. This is indeed the case, as the following lemma demonstrates.
\begin{lemma}
	\label{lem:commutant}
	We have that $\alg{A}(\Lambda^c)' \mc{H}_\Lambda \subset \mc{H}_\Lambda$.
\end{lemma}
\begin{proof}
	Let $B' \in \alg{A}(\Lambda^c)'$. Suppose $\zeta = F_1 \cdots F_n \Omega$ with $F_i \in \mc{F}_\Lambda$ and let $\eta = \widehat{F}_1 \cdots \widehat{F}_k F \Omega$, where $\widehat{F}_i \in \mc{F}_{\Lambda^c}$ and $F$ is a product of operators in $\mc{F}_\Lambda$. We will show that $(\eta, B'\zeta) = 0$ if $\eta \in \mc{H}_\Lambda^\perp$. Since the span of such vectors $\zeta$ (resp. $\eta$) is dense in $\mc{H}_\Lambda$ (resp. $\mc{H}$), the claim will follow. Now suppose that there is star $s$ such that $s \subset \Lambda^c$ and such that $A_s$ anti-commutes with $\widehat{F}_1 \cdots \widehat{F}_k$. Then, by locality and equation~\eqref{eq:stplaq},
\[
	(\eta, B'\zeta) = (\eta, B' A_s \zeta) = (A_s \eta, B' \zeta) = -(\eta, B'\zeta),
\]
hence $\eta$ is orthogonal to $B'\zeta$. A similar argument works for plaquette operators $B_p \in \alg{A}(\Lambda^c)$.

The case remains where no such plaquette or star operator exists. We claim that in this case, in fact $\eta \in \mc{H}_\Lambda$. First of all, note that any loops formed by the paths $\widehat{\xi}_i$ (corresponding to $\widehat{F}_i$) can be eliminated. Indeed, if $\xi_1, \dots \xi_k$ forms a loop, then $\widehat{F}_1 \cdots \widehat{F}_k$ is a product of either star or plaquette operators (see the end of Section~\ref{sec:model}). By commuting them with the other operators, and using equation~\eqref{eq:stplaq}, these can be eliminated, possibly at the expense of an overall minus sign. Similarly, if some of the paths $\widehat{\xi}_i$ can be combined to a bigger path, we might as well replace the string operators with the string operator of the bigger path.

Arguing like this, without loss of generality we can assume that the $\widehat{F}_i$ all correspond to different paths with mutually disjoint endpoints. It follows that the star and plaquette operators based at these endpoints anti-commute with $\widehat{F}_1 \cdots \widehat{F}_k$. By the assumption on $\eta$, this implies that all endpoints must lie on the boundary of $\Lambda$. So suppose that $\widehat{\xi}_i$ is a path with endpoints on the boundary of $\Lambda$. Then there is a path $\xi_i'$ inside $\Lambda$ with the same endpoints. If $F_{i'}$ is the corresponding string operator, then $\widehat{F}_i \Omega = F_{i'} \Omega$. Continuing in this manner, it follows that $\eta = F F_{k'} \cdots F_{1'} \Omega$. Hence $\eta \in \mc{H}_\Lambda$, completing the proof. 
\end{proof}

Since the lemma implies that $P_\Lambda \in \alg{A}(\Lambda^c)''$, we obtain the following corollary.
\begin{corollary}
	\label{cor:projection}
The projection $P_\Lambda$ is contained in $\mc{R}_{\Lambda^c}$.
\end{corollary}

We now consider $*$-algebras $\mc{A}_\Lambda$ and $\mc{B}_\Lambda$ acting on $\mc{H}_\Lambda$. Any operator $A \in \alg{A}(\Lambda)''$ restricts to an operator on $\mc{H}_\Lambda$ by Lemma~\ref{lem:unique}. Define an algebra $\mc{A}_\Lambda$ by restricting the operators of $\alg{A}(\Lambda)''$ to $\mc{H}_\Lambda$. This is in fact a von Neumann algebra, that is, $\mc{A}_\Lambda = \mc{A}_\Lambda''$ (as subalgebras of $\alg{B}(\mc{H}_\Lambda)$). This can be argued, for example, as in the proof of Prop. II.3.10 of Ref.~\cite{MR1873025}.

The algebra $\mc{B}_\Lambda$ is defined in a similar way: the operators in $P_\Lambda \mc{R}_{\Lambda^c} P_\Lambda$ leave $\mc{H}_\Lambda$ invariant, hence we can restrict $P_\Lambda \mc{R}_{\Lambda^c} P_\Lambda$ to a $*$-algebra acting on $\mc{H}_\Lambda$. This algebra will be denoted by $\mc{B}_\Lambda$ and is a von Neumann algebra by the proposition cited above. Note that both $\mc{A}_\Lambda$ and $\mc{B}_\Lambda$ act non-degenerately on $\mc{H}_\Lambda$ and that $\Omega$ is cyclic for $\mc{A}_\Lambda$.\footnote{In fact, one can show that $\Omega$ is separating for $\mc{B}_\Lambda$, but we will not need this fact.}  The self-adjoint part of $\mc{A}_\Lambda$ (resp. $\mc{B}_\Lambda$) is denoted by $\mc{A}_{\Lambda,s}$ (resp. $\mc{B}_{\Lambda,s})$. The following Lemma is the crucial step in the proof of Haag duality. 
\begin{lemma}
	\label{lem:dense}
	The set $\mc{A}_{\Lambda,s} \Omega + i \mc{B}_{\Lambda,s} \Omega$ is dense in $\mc{H}_\Lambda$.
\end{lemma}
\begin{proof}
	First we observe that since $\mc{A}_s$ and $\mc{B}_s$ are real vector spaces, it is sufficient to show that vectors of the form $F\Omega$ and $i F \Omega$, where $F$ is a product of operators in $\mc{F}_\Lambda$, are contained in $\mc{A}_{\Lambda,s} \Omega + i \mc{B}_{\Lambda,s} \Omega$. So suppose that $F = F_1 \cdots F_n$ with $F_i \in \mc{F}_\Lambda$. Note that $F_i^* = F_i$, and that $F_i, F_j$ either commute or anti-commute. But this means that $F^* = \pm F$. If $F^* = F$, clearly $F \in \mc{A}_{\Lambda,s}$. In the other case  $i F$ is self-adjoint, hence $i F \in \mc{A}_{\Lambda,s}$. 

	Now suppose that there is either a star operator $A_s \in \mc{A}_\Lambda$ or a plaquette operator $B_p \in \mc{A}_\Lambda$ that anti-commutes with $F$. In the case that $F = F^*$, it follows that $i A_s F$ (or $i B_p F$) is self-adjoint. But $i A_s F \Omega = - i F A_s \Omega = -i F \Omega$, so that we can obtain real linear combinations of $i F \Omega$. In the case that $F^* = -F$, one can use the fact that $A_s F$ is self-adjoint to obtain \emph{real} multiples of $F \Omega$. Combining these results, we obtain vectors of the form $\lambda F \Omega$, with $\lambda \in \mathbb{C}$.

	One issue remains: operators $A_s$ or $B_p$ (contained in $\mc{A}_\Lambda$) that anti-commute with $F$ need not exist. But if this is the case, then $F\Omega$ can only have excitations at the boundary of $\Lambda$, by the same reasoning as in the proof of Lemma~\ref{lem:commutant}. By the same proof, note that there is $\widehat{F} \in \mc{B}_\Lambda$ such that $\widehat{F}\Omega = F \Omega$. One also sees that if $F = F^*$, then also $\widehat{F} = \widehat{F}^*$, arguing as follows. Let $F_1, F_2$ be the string operators corresponding to paths $\xi_1, \xi_2$ in $\Lambda$, with endpoints at the boundary of $\Lambda$. Now choose corresponding paths $\xi_1'$ and $\xi_2'$ in $\Lambda^c$ with path operators $F_{1'}$ and $F_{2'}$. If the paths $\xi_1, \xi_2$ are of the same type, $F_1$ and $F_2$ commute, and so will $F_{1'}$ and $F_{2'}$. If they are of different type, they commute if and only if $\xi_1$ and $\xi_2$ intersect an \emph{even} number of times. Otherwise they will anti-commute. Note that $\xi_1 \cup \xi_1'$ is a loop, and similarly for $\xi_2 \cup \xi_2'$. But a loop on the lattice and a loop on the dual lattice always intersect an \emph{even} number of times. From this it follows that if $\xi_1$ and $\xi_2$ intersect an even (odd) number of times, the same is true for $\xi_1'$ and $\xi_2'$. It follows that $F_1$ and $F_2$ (anti-)commute if and only if $F_{1'}$ and $F_{2'}$ do so. In other words, if $F_1 F_2$ (resp. $i F_1 F_2$) is self-adjoint, then so is $F_{1'} F_{2'}$ (resp. $i F_{1'} F_{2'}$). Continuing in this way, it is clear that complex multiples of $F\Omega$ are contained in $\mc{A}_{\Lambda,s} \Omega + i \mc{B}_{\Lambda,s} \Omega$, which finishes the proof.
\end{proof}

We are now in a position to prove the main theorem.
\begin{proof}[Proof of Theorem~\ref{thm:dual}]
	As was mentioned before, using locality one obtains the inclusion $\pi(\alg{A}(\Lambda))'' \subset \pi(\alg{A}(\Lambda^c))'$. To prove the reverse inclusion, we first note that $\mc{A}_\Lambda$ and $\mc{B}_\Lambda'$  generate each other's commutant (in $\alg{B}(\mc{H}_\Lambda))$, by Lemma~\ref{lem:dense} and a result of Rieffel and van Daele~\cite[Thm. 2]{MR0383096}, which says in fact that the claim on the commutants is equivalent to the statement in Lemma~\ref{lem:dense}. In other words, $\mc{A}_\Lambda = \mc{B}_\Lambda'$ as von Neumann algebras acting on $\mc{H}_\Lambda$.

	In order to prove $\pi(\alg{A}(\Lambda^c))' \subset \pi(\alg{A}(\Lambda))''$, first note that $\mc{B}_\Lambda$ is the reduced von Neumann algebra $(\mc{R}_{\Lambda^c})_{P_\Lambda}$, obtained by restricting $P_\Lambda \mc{R}_{\Lambda^c} P_\Lambda$ to $\mc{H}_\Lambda$. Consider an element $B' \in \mc{R}_{\Lambda^c}'$. By~\cite[Prop. II.3.10]{MR1873025}, the commutant of $\mc{B}_\Lambda$ is equal to $\mc{R}_{\Lambda^c}'$ restricted to $\mc{H}_\Lambda$. Write $B_\Lambda'$ for the restriction of $B'$ to $\mc{H}_\Lambda$. Then $B_\Lambda' \in \mc{B}_\Lambda' = \mc{A}_\Lambda'' = \mc{A}_\Lambda$. By Lemma~\ref{lem:unique} and the remarks following Corollary~\ref{cor:projection}, there is a unique $\widehat{A} \in \mc{R}_\Lambda$ such that $\widehat{A}|_{\mc{H}_\Lambda} = B'_\Lambda$. Let $\xi= \widehat{F}F\Omega \in \mc{H}$, where $\widehat{F}$ (resp. $F$) is a product of operators in $\mc{F}_{\Lambda^c}$ (resp. $\mc{F}_\Lambda$). Then
\[
B' \xi = \widehat{F} B' F \Omega = \widehat{F} B'_{\Lambda} F \Omega = \widehat{F} \widehat{A} F \Omega = \widehat{A} \widehat{F} F \Omega = \widehat{A}\xi,
\]
so that $\widehat{A} = B'$ and hence $B' \in \pi(\alg{A}(\Lambda))'' = \mc{R}_\Lambda$.
\end{proof}

\section{Distal split property}
If $\Lambda$ is a cone, the von Neumann algebra $\mc{R}_\Lambda$ is a factor of Type $\rm{II}_\infty$ or Type III~\cite{toricendo}. If we have two cones $\Lambda_1 \subset \Lambda_2$, then clearly $\mc{R}_{\Lambda_1} \subset \mc{R}_{\Lambda_2}$. The \emph{distal split property} then says that if the boundaries of the cones $\Lambda_1$ and $\Lambda_2$ are well separated, then there is in fact a Type I factor $\mc{N}$ sitting between these two algebras, $\mc{R}_{\Lambda_1} \subset \mc{N} \subset \mc{R}_{\Lambda_2}$. To make this precise, we recall the following definition~\cite{toricendo}:
\begin{definition}
For two cones $\Lambda_1 \subset \Lambda_2$, write $\Lambda_1 \ll \Lambda_2$ if any star or plaquette in $\Lambda_1 \cup \Lambda_2^c$ is either contained in $\Lambda_1$ or in $\Lambda_2^c$. We say that $\omega$ satisfies the \emph{distal split property for cones} if for any pair of cones $\Lambda_1 \ll \Lambda_2$ there is a Type I factor $\mc{N}$ such that $\mc{R}_{\Lambda_1} \subset \mc{N} \subset \mc{R}_{\Lambda_2}$.
\end{definition}

For the toric code model we are considering, the distal split property in fact follows from Haag duality~\cite[Thm. 5.2]{toricendo}. Here we give another, more direct proof. For the remainder of this section, fix two cones $\Lambda_1 \ll \Lambda_2$. The idea is to use a unitary operator $U$ to write $\mc{H}$ as a tensor product of three Hilbert spaces, in such a way that $U \mc{R}_{\Lambda_1} U^*$ acts on the first tensor factor. Similarly, $U \mc{R}_{\Lambda_2^c} U^*$ acts on the second tensor factor, and from this one can find an interpolating Type I factor. 

There is some redundancy in the description of the Hilbert space $\mc{H}$ as the linear span of vectors obtained by acting with path operators on the ground state vector $\Omega$. For example, as mentioned before, $F_{\xi_1} \Omega = F_{\xi_2} \Omega$ if $\xi_1$ and $\xi_2$ are paths with the same endpoints. This is rather inconvenient when defining operators acting on $\mc{H}$, and therefore we will find a more economical description.

To achieve this, we will have to choose certain paths in $\Lambda_0 := \bonds \setminus (\Lambda_1 \cup \Lambda_2^c)$. Note that this set is non-empty, since $\Lambda_1 \ll \Lambda_2$. Choose a point in the lattice on the boundary of $\Lambda_1$, one on the boundary of $\Lambda_2$, and a path $\xi_1^b \subset \Lambda_0$ between these points. Similarly, choose plaquettes on the boundary of $\Lambda_1$, respectively $\Lambda_2$, and a dual path $\xi_2^b \subset \Lambda_0$ between these plaquettes. Label the vertices and plaquettes in the interior of $\Lambda_0$ (i.e. those vertices and plaquettes not on the boundary of $\Lambda_1$ or $\Lambda_2^c$) by a set $I$. If $I$ is non-empty, fix a vertex $v$ and a plaquette $p$ in $I$. Let $\xi_v$ and $\xi_p$ be paths in $\Lambda_0$ from $v$ (resp. $p$) to the boundary of $\Lambda_1$. For each $i \in I\setminus\{v,p\}$, choose a path inside $\Lambda_0$ from $i$ to either $v$ or $p$. Thus we have obtained a collection $\Gamma := \{\xi_1^b, \xi_2^b \} \cup \{\xi_i : i \in I\}$ of paths. For each $\xi \in \Gamma$ there is the corresponding path operator $\widehat{F}_\xi$.

\begin{definition}
	Let $\{ \widehat{F}_\xi \}_{\xi \in \Gamma}$ be as above and set $\alg{F}_0 = \{ F_{\xi_1} \cdots F_{\xi_k} : \xi_i \in \Gamma\}$. The Hilbert space $\mc{H}_0$ is defined as the closure of $\linspan \alg{F}_0 \Omega$.
\end{definition}

The  dimension of $\mc{H}_0$ depends on the number of stars and plaquettes there are in the region $\Lambda_2 \cap \Lambda_1^c$. In general this means that $\mc{H}_0$ is infinite dimensional. However, one can consider, for example, a cone $\Lambda_2$ based in the origin and bounded by the lines $y = x$ and $y = -x$ (any of the four possibilities will do). If one chooses $\Lambda_1$ to be the cone with parallel edges such that the distance between the two apexes is one, then $\Lambda_1 \ll \Lambda_2$ and $\Lambda_2 \cap \Lambda_1^c$ contains no stars or plaquettes. In this case, $\mc{H}_0$ is finite-dimensional: $\alg{F}_0$ consists of $I$ and the operators corresponding to the chosen path and dual path (and their product). Hence $\mc{H}_0$ has dimension four.

The construction of $\mc{H}_0$ is perhaps somewhat involved, but it suggests a convenient description of $\mc{H}$. Analogously to $\alg{F}_0$, we define the set $\alg{F}_{\Lambda_1}$ by $\alg{F}_{\Lambda_1} = \{ F_1 \cdots F_n : F_i \in \mc{F}_{\Lambda_1} \}$ and in the same way $\alg{F}_{\Lambda_2^c}$.
\begin{lemma}
The set $\linspan \alg{F}_{\Lambda_1} \alg{F}_0 \alg{F}_{\Lambda_2^c} \Omega$ is dense in $\mc{H}$.
	\label{eq:canform}
\end{lemma}
\begin{proof}
	By Lemma~\ref{lem:denseset}, vectors of the form $F_{\xi_1} \cdots F_{\xi_n} \Omega$ span a dense subset of $\mc{H}$. Note that we can permute the order of the operators $F_{\xi_i}$, possibly at the expense of an overall sign. But this implies that it is enough to show that for a path $\xi$, $F_{\xi}\Omega$ is of the desired form. Suppose for the sake of argument that $\xi$ is a path on the lattice. If both endpoints of the path --- call them $v_1$ and $v_2$ --- are in either $\Lambda_1$ or $\Lambda_2^c$, the claim is clear. If $v_1$ is in $\Lambda_0$ and $v_2$ in $\Lambda_1$ or $\Lambda_2^c$, consider the path $\xi_{v_1} \cup \xi_v$ from $v_1$ to the boundary of $\Lambda_1$. If $v_2$ is in $\Lambda_1$, choose a path $\tilde{\xi}$ from this boundary point to $v_2$. Then we have $F_{\xi} \Omega = \widehat{F}_{\xi_{v_1}} \widehat{F}_{\xi_v} F_{\tilde{\xi}} \Omega$, which is of the desired form. If $v_2$ is in $\Lambda_2^c$ then one can form the following path: first go from $v_1$ to the boundary of $\Lambda_1$ as above. Then choose a path in $\Lambda_1$ from the endpoint of $\xi_{v}$ to the endpoint of either $\xi_1^b$ or $\xi_2^b$ and use this path to go to $\Lambda_2^c$. From there one can choose a path from the boundary to $v_2$ and we are done. The remaining cases can be handled in a similar way.
\end{proof}

The proof actually implies that every vector of the form $F_{\xi_1} \cdots F_{\xi_n} \Omega$ can be written (up to an overall sign) as $F_1 \widehat{F} F_2 \Omega$. We say that a vector is in \emph{canonical form} if it is represented in this way. The point is that some of the redundancy in the description is removed: if $F_1 \widehat{F} F_2 \Omega = \pm F_1' \widehat{F}' F_2' \Omega$ for $F_1,F_1' \in \alg{F}_{\Lambda_1}, F_2,F_2' \in \alg{F}_{\Lambda_2^c}$ and $\widehat{F}, \widehat{F}' \in \alg{F}_0$ then in fact $\widehat{F} = \pm \widehat{F}'$.
\begin{lemma}
	Suppose that $\Lambda_1 \ll \Lambda_2$ are two cones. If $F_1 \widehat{F} F_2 \Omega$ is in canonical form, define
\begin{equation}
	U F_1 \widehat{F} F_2 \Omega = F_1 \Omega \otimes F_2 \Omega \otimes \widehat{F} \Omega.
	\label{eq:unitary}
\end{equation}
Then $U$ extends to a unitary operator $\mc{H} \to \mc{H}_{\Lambda_1} \otimes \mc{H}_{\Lambda_2^c} \otimes \mc{H}_0$, where $\mc{H}_{\Lambda_1}, \mc{H}_{\Lambda_2^c}$, and $\mc{H}_0$ are the Hilbert spaces defined above.
\end{lemma}
\begin{proof}
	We first prove that $U$ defines an isometry, from which it is clear that $U$ is well-defined. Suppose that $\eta_1 = F_1 \widehat{F} F_2 \Omega$ and $\eta_2 = F_1' \widehat{F}' F_2' \Omega$ are in canonical form. It is enough to show that $(\eta_1, \eta_2) = (U \eta_1, U \eta_2)$. First suppose that $\widehat{F} \neq \pm \widehat{F}'$. Then there is some star or plaquette operator that commutes with $\widehat{F}$, but anti-commutes with $\widehat{F}'$ (or vice-versa), hence $\omega(\widehat{F}^*\widehat{F}) = 0$, and therefore $(U \eta_1, U\eta_2) = 0$. We claim that in this case $(\eta_1, \eta_2) = 0$. If there is a vertex or plaquette in the interior of $\Lambda_0$ where $\widehat{F}$ creates an excitation but $\widehat{F}$ doesn't (or vice versa), this equality is clear since then there is a star (or plaquette) operator that commutes with $\mc{R}_{\Lambda_1}$ and $\mc{R}_{\Lambda_2^c}$, but anti-commutes with either $\widehat{F}$ or $\widehat{F}'$. So suppose that this is not the case. Then $F_{\xi_1^b}$ or $F_{\xi_2^b}$ is necessarily a factor in either $\widehat{F}$ or $\widehat{F}'$, say $\widehat{F}$. But then $F_1 \widehat{F} F_2 \Omega$ has an \emph{odd} number of excitations localized in $\Lambda_1$ or at its boundary. The same holds for $\Lambda_2^c$. On the other hand, $F_1' \widehat{F}' F_2' \Omega$ has an \emph{even} number of excitations it both these regions. So there must be at least one place where one vector has an excitation and the other one does not. But this implies that $(\eta_1, \eta_2) = 0$ as before.

Hence without loss of generality we can assume that $\widehat{F} = \widehat{F}'$ and the problem reduces to showing that $\omega(F_1^* F_1' F_2^* F_2') = \omega(F_1^* F_1') \omega(F_2^* F_2')$. This equality can be obtained as follows: if there is a star or plaquette operator that anti-commutes with any of the operator $F_i, F_i'$ and commutes with the others, both sides are zero by the same reasoning as used before. If this is not the case, this implies that $F_1^* F_2$ and $\widehat{F}_1^*\widehat{F}_2$ correspond to products of path operators of closed loops, and it follows that both sides are equal to plus or minus one. The sign has to be equal at both sides, since $F_1,F_1'$ and $F_2,F_2'$ commute. The range of $U$ is clearly dense in $\mc{H}_{\Lambda_1} \otimes \mc{H}_{\Lambda_2^c} \otimes \mc{H}_0$, hence $U$ extends to a unitary operator.
\end{proof}

This unitary gives the desired decomposition of $\mc{H}$ as a tensor product of Hilbert spaces. The proof of the main theorem of this section now amounts to showing that $\mc{R}_{\Lambda_1}$ and $\mc{R}_{\Lambda_2}$ act on this tensor product in the desired way. 
\begin{theorem}
Suppose that $\Lambda_1 \ll \Lambda_2$ and let $U$ be the unitary defined as above. If $\mc{N} = U^* \left( \alg{B}(\mc{H}_{\Lambda_1}) \otimes I \otimes I \right)U$, then $\mc{N}$ is a Type I factor such that $\mc{R}_{\Lambda_1} \subset \mc{N} \subset \mc{R}_{\Lambda_2}$.
\end{theorem}
\begin{proof}
	It is clear that $\mc{N}$ is a Type I factor, hence it remains to show the inclusions. We will show that $U \mc{R}_{\Lambda_1} U^* = \mc{R}_{\Lambda_1} P_{\Lambda_1} \otimes I \otimes I$ and similarly $U \mc{R}_{\Lambda_2}' U^* = I \otimes \mc{R}_{\Lambda_2}' P_{\Lambda_2^c} \otimes I$, where $\mc{R}_{\Lambda_1} P_{\Lambda_1}$ is the von Neumann algebra $\mc{R}_{\Lambda_1}$ restricted to $\mc{H}_{\Lambda_1}$. It follows that $\mc{R}_{\Lambda_1} \subset \mc{N}$. For the second inclusion, note that 
$$U \mc{R}_{\Lambda_2}'' U^* = (I \otimes \mc{R}_{\Lambda_2}' P_{\Lambda_2^c} \otimes I)' = \alg{B}(\mc{H}_{\Lambda_1}) \otimes P_{\Lambda_2^c} \mc{R}_{\Lambda_2}'' P_{\Lambda_2^c} \otimes \alg{B}(\mc{H}_0),$$
and hence $\mc{N} \subset \mc{R}_{\Lambda_2}'' = \mc{R}_{\Lambda_2}$.

Note that if $\eta \in \mc{H}_{\Lambda_1}$ and $F \in \mc{F}_{\Lambda_2^c}, \widehat{F} \in \mc{F}_0$ then $\widehat{F}F\eta \in \mc{H}$ and by definition $U \widehat{F} F \eta = \eta \otimes F\Omega \otimes \widehat{F}\Omega$ and similarly for $\eta \in \mc{H}_{\Lambda_2^c}$. To finish the proof, first recall that by Lemma~\ref{lem:unique}, $\mc{R}_{\Lambda_1} \mc{H}_{\Lambda_1} \subset \mc{H}_{\Lambda_1}$. In a similar way one shows that $\mc{R}_{\Lambda_2}' = \mc{R}_{\Lambda_2^c}$ maps $\mc{H}_{\Lambda_2^c}$ into itself. Now, suppose that $A \in \mc{R}_{\Lambda_1}$ and $\eta := F_1\Omega \otimes F_2 \Omega \otimes \widehat{F}\Omega \in \mc{H}_{\Lambda_1} \otimes \mc{H}_{\Lambda_2^c} \otimes \mc{H}_0$. By locality $A$ commutes with $F_2$ and $\widehat{F}$. One then finds
\[
\begin{split}
U A U^* &\eta = U A F_1 \widehat{F} F_2 \Omega = U \widehat{F} F_2 A F_1 \Omega = U \widehat{F} F_2 P_{\Lambda_1} A P_{\Lambda_1} F_1 \Omega\\
&= A|_{\Lambda_1} F_1 \Omega \otimes F_2 \Omega \otimes \widehat{F} \Omega = \left( A|_{\Lambda_1} \otimes I \otimes I\right) \eta.
\end{split}
\]
Since vectors of the form $\eta$ span a dense set, the claim for $U \mc{R}_{\Lambda_1} U^*$ follows. A similar argument then shows the corresponding claim for $\mc{R}_{\Lambda_2}'$, which concludes the proof.
\end{proof}
One can in fact set $\mc{N}_1 := \mc{N}$ and $\mc{N}_2 := U^* (\alg{B}(\mc{H}_{\Lambda_1}) \otimes I \otimes \alg{B}(\mc{H}_0))U$ and it follows that $\mc{R}_{\Lambda_1} \subset \mc{N}_1 \subset \mc{N}_2 \subset \mc{R}_{\Lambda_2}$. This inclusion of \emph{two} Type I factors is also found in the case of the free neutral massive scalar field in algebraic quantum field theory, discussed by Buchholz~\cite[Corr. 2.4]{MR0345546}.

Note that in the case that $\mc{R}_{\Lambda_1}$ and $\mc{R}_{\Lambda_2}$ are semi-finite, the construction here is an explicit example of the construction in the proof of~\cite[Cor. 1(iv)]{MR703083}. Indeed, consider $\mc{R}_{\Lambda_1} \otimes \mc{R}_{\Lambda_2}'$. Then there is an amplification $\mc{R}_{\Lambda_1} \otimes \mc{R}_{\Lambda_2}' \otimes I$ acting on the Hilbert space $\mc{H} \otimes \mc{H} \otimes \mc{H}$. Let $P_0$ be the projection onto $\mc{H}_0$. If one reduces the amplification by the projection $P_{\Lambda_1} \otimes P_{\Lambda_2^c} \otimes P_0 \in \mc{R}_{\Lambda_1}' \otimes \mc{R}_{\Lambda_2} \otimes \alg{B}(\mc{H})$ and conjugates with the unitary $U$, one obtains a normal faithful representation of $\mc{R}_{\Lambda_1} \otimes \mc{R}_{\Lambda_2}'$ onto $\mc{R}_{\Lambda_1} \vee \mc{R}_{\Lambda_2}'$.

\vspace{\baselineskip}
\noindent\emph{Acknowledgements:} This research is funded by the Netherlands Organisation for Scientific Research (NWO) grant no.~613.000.608. The author wishes to thank Klaas Landsman, Michael M\"uger and Reinhard Werner for valuable feedback on the manuscript.

\bibliographystyle{abbrv}
\bibliography{refs.bib}
\end{document}